\DocumentMetadata{}
\documentclass[sigconf]{acmart}
\AtBeginDocument{%
  }

\setcopyright{acmcopyright}
\copyrightyear{2026}
\acmYear{2026}
\acmConference[CIDR '26]{Make sure to enter the correct
  conference title from your rights confirmation email}{January 18--21,
  2026}{Chaminade, USA}

\usepackage{ebproof}
\usepackage{algorithm}
\usepackage{algpseudocode}

\newtheorem{lem}{Lemma}[section]
\newtheorem{thm}{Theorem}[section]

\begin{document}

\title{Towards General-Purpose Data Discovery: \\ A Programming Languages Approach}

\author{Andrew Y. Kang}
\affiliation{%
  \institution{Cornell University}
  \city{Ithaca}
  \state{New York}
  \country{USA}
}
\email{ayk36@cornell.edu}

\author{Yashnil Saha}
\affiliation{%
  \institution{Cornell University}
  \city{Ithaca}
  \state{New York}
  \country{USA}
}
\email{ys2343@cornell.edu}

\author{Sainyam Galhotra}
\affiliation{%
  \institution{Cornell University}
  \city{Ithaca}
  \state{New York}
  \country{USA}
}
\email{sg@cs.cornell.edu}

\renewcommand{\shortauthors}{Kang et al.}

\begin{abstract}
    Efficient and effective data discovery is critical for many modern applications in machine learning and data science. One major bottleneck to the development of a general-purpose data discovery tool is the absence of an expressive formal language, and corresponding implementation, for characterizing and solving generic discovery queries. To this end, we present TQL, a domain-specific language for data discovery well-designed to leverage and exploit the results of programming languages research in both its syntax and semantics. In this paper, we fully and formally characterize the core language through an algebraic model, Imperative Relational Algebra with Types (\texttt{ImpRAT}), and implement a modular proof-of-concept system prototype.
\end{abstract}



\received{5 August 2025}

\maketitle

\section{Introduction}
Modern applications in machine learning and data science rely upon easy access to large quantities of high-quality data. Efficiently and effectively leveraging the resources of very large data repositories is therefore essential for the practical building and scaling of deployable models. However, due to the enormity of dataset search spaces \cite{Hulsebos:hard}, and the substantial variance in dataset quality, successfully realizing  data discovery is often difficult and labor-intensive.

These data discovery challenges are rarely the result of poor tool or system design. Indeed, there is remarkable research on the development and implementation of tools and systems for particular critical scenarios in data discovery \cite{Galhotra:metam, tao:kyrix}. The more fundamental driving difficulty of general-purpose data discovery, in our view, is the absence of a formal and standardized approach for characterizing the dynamic and diverse requests which arise in data science and machine learning applications. 

Our preliminary work~\cite{kang:tql1} tackled this "language problem" in general-purpose data discovery with an inter- and cross-disciplinary approach \cite{wu:acdemic}. We proposed TQL, a data discovery language using the observations that most meaningful data discovery queries could be decomposed into two more primitive classes of queries: transformation clauses, which characterize the manner in which raw data is modified; and constraint clauses, which describe properties that the data is expected to satisfy. Our main insight then, extending philosophically from the Curry-Howard Isomorphism \cite{Wadler:propos}, was to identify that transformation and constraint clauses shared behavior analogous with that of operations and types in a typical programming language --- thus suggesting a natural encoding of the former into the latter. 

We proposed TQL, an easily extensible data discovery language for describing meaningful problems in data discovery, and produced an associated set-theoretic formalization. This paper builds on our initial work by articulating the necessary mathematical foundations (through the algebraic \texttt{ImpRAT} model) to fully formalize the operational semantics of the TQL language, and by introducing and implementing the core architecture for a proof-of-concept prototype solver for language-describable data discovery problems.

\section{Language Overview}
This section overviews the formal syntax and semantics for the TQL language.  We proceed through this section in three parts. We begin by providing some intuition through a review of a set-theoretic model for our approach to general-purpose data discovery: Type-Collected Relational Algebra (TCRA). 
 We then introduce Imperative Relational Algebra with Types (\texttt{ImpRAT}), a largely algebraic language that provides the algorithmic foundations for our data discovery language. 
We finally give the formal semantics for TQL via a translation from the surface syntax of TQL into the algebraic syntax of \texttt{ImpRAT}.

\subsection{Type-Collected Relational Algebra (TCRA)}

We summarize our constructed set-theoretic model for TQL, Type-Collected Relational Algebra (TCRA) \cite{kang:tql1}. This model both provides some basic intuition for our approach, and gives an implicit proof that our formulation of the data discovery problem was both computable and well-defined. 

At a high-level, type-collected relational algebra can be viewed as an extension of traditional relational algebra where relational operations act upon typed dataset variables in addition to explicitly declared datasets. 

These ``dataset variables", representing undetermined datasets in the query statement, are modeled as sets of all known candidate datasets which might be substituted for the variable (i.e. a subset of the finite set $\mathcal{D}$ of all datasets in the data repository). Each dataset variable is "typed" to the extent that it has an associated predicate denoting particular properties that its contained datasets must satisfy.  

Accordingly, relational algebra operations are redefined as their corresponding set-level operations. For example, if $\tau(t_0, t_1)$ is an operation in relational algebra, then the application of $\tau$ to dataset variables $T_1$ and $T_2$ is given by: 
\begin{displaymath}
    \tau(T_1, T_2) \equiv \{ \tau(t_1, t_2): t_1 \in T_1 \land t_2 \in T_2\}
\end{displaymath}
For consistency with this redefined relational operation, explicitly declared datasets $t_0$ should be viewed as a singleton dataset variable (formally, the set $\{ t_0 \}$).

The TCRA model is advantageous insofar as it can capture all valid paths of a data discovery query. That is, if a dataset is in the resultant set (of datasets) of a query, then it must be producible by a certain substitution of the dataset variables.

We now recover a desirable property of TCRA. Specifically, we show that TCRA is relationally complete. 

\begin{lem}[Natural Inclusion] There is a natural inclusion mapping, given by the function $\imath(t) = \{ t \}$, from relational algebra into type-collected relational algebra. 
\end{lem}

\begin{proof}
    It suffices to show that operations in type-collected relational algebra are closed on single-element collections for all tabular values on which they are designed. 
    
    Observe that by definition, each table in the output set of a type-collected relational algebra operation is constructed from selecting a single table from each input set to serve as the input tables for the associated table-level function. As such, it follows combinatorially that if each input set is a singleton, then the output set must also be a singleton. And so, closure holds. 
\end{proof}

\begin{thm}[Relational Completeness]
     Type-Collected Relational Algebra (TCRA) is relationally complete. That is, it is a valid extension of Relational Algebra (RA).
\end{thm}

\begin{proof}
    This follows immediately from Lemma 2.1. 
\end{proof}

\subsection{Imperative Relational Algebra with Types (\texttt{ImpRAT})}

For the purposes of language implementation, we now present an alternative, but essentially equivalent, formulation for TQL through the algebraic language of \texttt{ImpRAT}. Unlike our initial set-theoretic formulation, a given expression in \texttt{ImpRAT} does not entirely describe the behavior of a data discovery "program". Instead, a program is given as a tuple $(S, \mathcal{D})$ consisting of both an \texttt{ImpRAT} program $S$, and the set $\mathcal{D}$ of all datasets $D$ in a data repository. A resultant dataset $D'$ for a discovery program $(S, \mathcal{D})$ is defined as any valid (i.e. non-NULL) result from computing the output of program $S$ on a finite tuple of not-necessarily distinct input datasets $D_1, D_2, ..., D_n \in \mathcal{D}$. Formally, if we view program $S$ as a function $S : \mathcal{D}^n \rightarrow \mathcal{D}_* \cup \{ \bot \}$, where $\mathcal{D}_*$ is the closure of $\mathcal{D}$ on all dataset operations in relational algebra, then the set of valid inputs can be characterized as:
\begin{displaymath}
  S^{-1}(\mathcal{D}_*) \equiv \{ (D_1, ..., D_n) \in \mathcal{D}^n : S(D_1, ..., D_n) \neq \bot \}
\end{displaymath}
and the resultant datasets $D'$ can be represented as: 
\begin{displaymath}
  (S \circ S^{-1})(\mathcal{D}_*) \equiv \{ S(D_1, ..., D_n) : (D_1, ..., D_n) \in S^{-1}(\mathcal{D}_*) \}
\end{displaymath}

Intuitively, within this more computational formulation of the data discovery problem, \texttt{ImpRAT} can be understood as a descriptive language for modeling the transformations and constraints which the input datasets and their resultands are expected to satisfy. 

In designing an algebraic language for such a purpose, we model our approach on similar work in the domain of network programming languages \cite{anderson:netkat}. Specifically, we take inspiration from the construction of a Kleene Algebra with Tests \cite{kozen:kat}, embedding simple boolean logic into relational algebra \cite{aho:universality}. This simple extension, which we call "Relational Algebra with Types" (RAT), allows relational algebra to check for desirable constraints on queried datasets. We supplement \texttt{RAT} with a framework of basic imperative constructs \cite{winskel:formal}, creating \texttt{ImpRAT}, to equip the language with a simple control flow that is easy to reason about. 

Syntactically, \texttt{ImpRAT} expressions are formed from three categories: relational expressions, types, and statements. Relational expressions include the typical operations in relational algebra, as well as the "test" operation $R \{ T_0, ..., T_n \}$ which checks the predicate type of a relation (i.e. a dynamically-evaluated type ascription). Types consist of predicate tests on basic properties of datasets, and the symbols of first-order logic. Statements follow the standard structure of a typical imperative language \cite{winskel:formal}, with $\Delta$ corresponding to a "return" operation. A special $!x$ operation allows the user to declare that a variable $x$ is free in program $S$ (i.e. that $x$ requires an input dataset assignment). The syntax for \texttt{ImpRAT} is given in complete detail as a BNF grammar in Figure ~\ref{fig:BNF-RAT}. 

\begin{figure}
    \caption{Syntax for ImpRAT}
    \label{fig:BNF-RAT}
    \begin{tabular}{l}
        \toprule
        \vbox{
            \begin{flalign*}
                R &::= D {\:|\:} x {\:|\:} R_0 \cup R_1 {\:|\:} R_0 - R_1 {\:|\:} R_0 \times R_1 {\:|\:} \bot \\ 
                &{\:|\:} \Pi_{a_1, ..., a_n} (R) {\:|\:} \sigma_\varphi(R) {\:|\:} \rho_{a/b}(R) {\:|\:} R \{ T_0, ..., T_n \} \\
                T &::= a  {\:|\:} \exists \varphi {\:|\:} \forall \varphi \\
                S &::= \top {\:|\:} \bot {\:|\:} x := R {\:|\:} !x {\:|\:} S_0 ; S_1 {\:|\:} \Delta R
            \end{flalign*}
        } \\
        \bottomrule
    \end{tabular}
\end{figure}

Semantically, an \texttt{ImpRAT} expression consists of a series of statements which either evaluate to a successful completion (i.e. evaluate to $\top$ or $\Delta D$), or do not (i.e. evaluate to $\bot$). An \texttt{ImpRAT} expression is considered to have evaluated to successful completion on the whole if and only if there is a statement that evaluates to $\Delta D$ before any statements evaluate to $\bot$. A valid \texttt{ImpRAT} expression is expected to only evaluate to either $\bot$ or $\Delta D$. A correct \texttt{ImpRAT} expression cannot evaluate to $\top$. The assignment ($x := R$) and return ($\Delta R$) statements are the basic building blocks of most meaningful \texttt{ImpRAT} expressions. These statements evaluate to $\top$ or $\Delta D$ respectively if and only if their relational expression $R$ does not evaluate to $\bot$. 

A relational expression $R$ is any legal expression in \texttt{RAT}, which is largely identical to basic relational algebra with the minor addition of a test operation $R \{ T_0, ..., T_n \}$, and the presence of dataset variables $x$, whose values are mapped by the environment $\kappa$. Relational expressions typically evaluate to the NULL character $\bot$ from failing a test operation. We note, however, that relational expressions may also evaluate to $\bot$ if an illegal operation is performed (e.g. attempting to perform a set-union on two datasets which are not union-compatible). For the sake of semantic consistency, we structure the language so that these cases are typically caught by induced instances of the test operation. 

To provide some additional intuition, we detail a few representative small-step semantic rules in both the 'relational expression' and 'statement' categories:

\textbf{UNION ($\cup$)}. The union operator is associated with two semantic rules: 
    \begin{displaymath}
        \begin{split}
        & \begin{prooftree}
            \hypo{$\,$}
            \infer1{\langle D_1 \cup D_2 , \kappa \rangle \rightarrow D_3} 
        \end{prooftree} \text{ \small (w/e $D_3$ is the union of relations)} 
        \\
        & \, \\
        & \begin{prooftree}
            \hypo{$\,$}
            \infer1{\langle D_1 \cup D_2 , \kappa \rangle \rightarrow \bot} 
            \end{prooftree} \text{ \small (if $D_1$ and $D_2$ are not union compatible)}
        \end{split}
    \end{displaymath}
If $D_1$ and $D_2$ are union-compatible (i.e. there is an appropriate matching of their columns), then the first rule applies, and returns a dataset $D_3$ which is equal to the union of $D_1$ and $D_2$ in relational algebra. If $D_1$ and $D_2$ are not union-compatible, then the second rule applies, and returns the NULL character $\bot$. 

\textbf{UNIVERSAL PREDICATE ($\forall \varphi$)}. The test operation for the universal predicate is associated with two semantic rules: 
    \begin{displaymath}
        \begin{split}
        & \begin{prooftree}
         \hypo{$\,$}
        \infer1{\langle D \{ \forall \varphi \} , \kappa \rangle \rightarrow D} 
        \end{prooftree} \text{\small (if all tuples in $D$ satisfy $\varphi$)}
        \\
        & \, \\
        & \begin{prooftree}
        \hypo{$\,$}
        \infer1{\langle D \{ \forall \varphi \} , \kappa \rangle \rightarrow \bot} 
        \end{prooftree} \text{\small (if any tuples in $D$ do not satisfy $\varphi$)} 
        \end{split}
    \end{displaymath}
If all rows in dataset $D$ satisfy the predicate test $\varphi$, then the first rule applies and $D$ passes the test operation. If any row in dataset $D$ does not satisfy the predicate test $\varphi$, then the second rule applies and the NULL character $\bot$ is returned.

\textbf{ASSIGNMENT ($:=$)}. The assignment statement is associated with three semantic rules: 
    \begin{displaymath}
        \begin{split}
            & \begin{prooftree}
                \hypo{\langle R, \kappa \rangle \rightarrow R'}
                \infer1{ \langle x := R, \kappa \rangle \rightarrow \langle x:= R', \kappa \rangle}
            \end{prooftree} \\
            & \, \\
            & \begin{prooftree}
                \hypo{$\,$}
                \infer1{ \langle x := D, \kappa \rangle \rightarrow \langle \top, \kappa [D/x] \rangle}
            \end{prooftree} \\
            & \, \\
            & \begin{prooftree}
                \hypo{$\,$}
                \infer1{ \langle x := \bot, \kappa \rangle \rightarrow \langle \bot, \kappa \rangle}
            \end{prooftree}
        \end{split}
    \end{displaymath}
The first rule states that an assignment cannot occur until the relational expression $R$ has finished evaluation. The second rule states that if the result of evaluating $R$ is a dataset $D$, then the variable environment $\kappa$ should be updated with $x = D$, and the program should continue onto the next statement. The third rule states that if the evaluation of $R$ fails, then the assignment statement should fail immediately as well. 

The full semantics for both relational expressions and statements are given as small-step operational semantics in the corresponding columns of Figure ~\ref{fig:SEM-RAT}. Small-step semantics were chosen over their equivalent big-step representation for purposes of greater clarity.

\begin{figure*}
  \caption{Semantics for ImpRAT}
  \label{fig:SEM-RAT}
  \def\arraystretch{2.3}
  \begin{tabular}{ll|l}
    \toprule

    \textbf{Relational Expressions} & & \textbf{Statements} \\


    \small
    \begin{prooftree}
        \hypo{$\,$}
        \infer1{\langle x, \kappa \rangle \rightarrow \kappa(x)}
    \end{prooftree} & &
      
    \small
    \begin{prooftree}
        \hypo{$\,$}
        \infer1{ \langle !x, \kappa \rangle \rightarrow \langle \top, \kappa \rangle}
    \end{prooftree} \\


    \small
    \begin{prooftree}
            \hypo{$\,$}
            \infer1{\langle D_1 \cup D_2 , \kappa \rangle \rightarrow D_3} 
    \end{prooftree} { \small (w/e $D_3$ is the union of relations)} & 

    \small
    \begin{prooftree}
            \hypo{$\,$}
            \infer1{\langle D_1 \cup D_2 , \kappa \rangle \rightarrow \bot} 
    \end{prooftree} { \small (if $D_1$ and $D_2$ are not union compatible)} & 

    \small
    \begin{prooftree}
        \hypo{$\,$}
        \infer1{ \langle x := D, \kappa \rangle \rightarrow \langle \top, \kappa [D/x] \rangle}
    \end{prooftree} \\


    \small
    \begin{prooftree}
        \hypo{$\,$}
        \infer1{\langle D_1 - D_2 , \kappa \rangle \rightarrow D_3} 
    \end{prooftree} {\small (w/e $D_3$ is the difference of relations)} &

    \small
    \begin{prooftree}
            \hypo{$\,$}
            \infer1{\langle D_1 - D_2 , \kappa \rangle \rightarrow \bot} 
        \end{prooftree} {\small (if $D_1$ and $D_2$ are not union compatible)} &

    \small
    \begin{prooftree}
        \hypo{$\,$}
        \infer1{ \langle x := \bot, \kappa \rangle \rightarrow \langle \bot, \kappa \rangle}
    \end{prooftree} \\


     \small
        \begin{prooftree}
            \hypo{$\,$}
            \infer1{\langle D_1 \times D_2 , \kappa \rangle \rightarrow D_3} 
        \end{prooftree} {\small (w/e $D_3$ is the product of relations)} & &

    \small
    \begin{prooftree}
        \hypo{$\,$}
        \infer1{ \langle \top; c_1, \kappa \rangle \rightarrow \langle c_1, \kappa \rangle}
    \end{prooftree} \\

    \small
    \begin{prooftree}
        \hypo{$\,$}
        \infer1{\langle D \{ \exists \varphi \} , \kappa \rangle \rightarrow D} 
    \end{prooftree} {\small (if any tuple in $D$ satisfies $\varphi$)} &

     \small
    \begin{prooftree}
        \hypo{$\,$}
        \infer1{\langle D \{ \exists \varphi \} , \kappa \rangle \rightarrow \bot} 
    \end{prooftree} {\small (if no tuple in $D$ satisfies $\varphi$)} &

    \small
    \begin{prooftree}
        \hypo{$\,$}
        \infer1{ \langle \bot; c_1, \kappa \rangle \rightarrow \langle \bot, \kappa \rangle}
    \end{prooftree} \\

    \small
    \begin{prooftree}
         \hypo{$\,$}
        \infer1{\langle D \{ \forall \varphi \} , \kappa \rangle \rightarrow D} 
    \end{prooftree} {\small (if all tuples in $D$ satisfy $\varphi$)} &

    \small
    \begin{prooftree}
        \hypo{$\,$}
        \infer1{\langle D \{ \forall \varphi \} , \kappa \rangle \rightarrow \bot} 
    \end{prooftree} {\small (if any tuples in $D$ do not satisfy $\varphi$)} &

    \small
    \begin{prooftree}
        \hypo{$\,$}
        \infer1{ \langle \Delta D; c_1, \kappa \rangle \rightarrow \langle \Delta D, \kappa \rangle}
    \end{prooftree} \\
    
    \small
    \begin{prooftree}
        \hypo{$\,$}
        \infer1{\langle D \{ a_1, ..., a_n \} , \kappa \rangle \rightarrow D} 
    \end{prooftree} {\small (if $D$ has attributes $a_1, ..., a_n$)} &

    \small
    \begin{prooftree}
        \hypo{$\,$}
        \infer1{\langle D \{ a_1, ..., a_n \} , \kappa \rangle \rightarrow \bot} 
    \end{prooftree} {\small (if $D$ lacks attributes $a_1, ..., a_n$)} &

    \small
    \begin{prooftree}
        \hypo{$\,$}
        \infer1{ \langle \Delta \bot,  \kappa \rangle \rightarrow \langle \bot, \kappa \rangle}
    \end{prooftree} \\
    
    \small
    \begin{prooftree}
        \hypo{ \langle D \{ a_1, ..., a_n \}, \kappa \rangle \rightarrow D }
        \infer1{\langle \Pi_{a_1, ..., a_n} (D) , \kappa \rangle \rightarrow D_1}
    \end{prooftree} { \small (w/e $D_1$ is a restriction to $\{ a_1, ..., a_n \}$)}  &

    \small
    \begin{prooftree}
        \hypo{ \langle D \{ a_1, ..., a_n \}, \kappa \rangle \rightarrow \bot }
        \infer1{\langle \Pi_{a_1, ..., a_n} (D) , \kappa \rangle \rightarrow \bot}
    \end{prooftree} &

    \small
    \begin{prooftree}
        \hypo{\langle R, \kappa \rangle \rightarrow R'}
        \infer1{ \langle \Delta R,  \kappa \rangle \rightarrow \langle \Delta R', \kappa \rangle}
    \end{prooftree} \\

    \small
    \begin{prooftree}
        \hypo{ \langle D \{ \exists \varphi \}, \kappa \rangle \rightarrow D }
        \infer1{\langle \sigma_{\varphi} (D) , \kappa \rangle \rightarrow D_1}
    \end{prooftree} { \small (w/e $D_1$ is the tuples of $D$ satisfying $\varphi$)} &

    \small
    \begin{prooftree}
        \hypo{ \langle D \{ \exists \varphi \}, \kappa \rangle \rightarrow \bot }
        \infer1{\langle \sigma_{\varphi} (D) , \kappa \rangle \rightarrow \bot}
    \end{prooftree} &

    \small
    \begin{prooftree}
        \hypo{\langle R, \kappa \rangle \rightarrow R'}
        \infer1{ \langle x := R, \kappa \rangle \rightarrow \langle x:= R', \kappa \rangle}
    \end{prooftree} \\

   
    \small
    \begin{prooftree}
        \hypo{ \langle D \{ b \}, \kappa \rangle \rightarrow D }
        \infer1{\langle \rho_{a / b} (D) , \kappa \rangle \rightarrow \bot}
    \end{prooftree} &

    \small
    \begin{prooftree}
        \hypo{ \langle D \{ a \}, \kappa \rangle \rightarrow \bot }
        \infer1{\langle \rho_{a / b} (D) , \kappa \rangle \rightarrow \bot}
    \end{prooftree} &
    
    \small
    \begin{prooftree}
        \hypo{\langle c_0, \kappa \rangle \rightarrow \langle c_0', \kappa' \rangle}
        \infer1{ \langle c_0; c_1, \kappa \rangle \rightarrow \langle c_0';c_1, \kappa' \rangle}
    \end{prooftree} \\


    \small
    \begin{prooftree}
        \hypo{ \langle D \{ a \}, \kappa \rangle \rightarrow D }
        \hypo{ \langle D \{ b \}, \kappa \rangle \rightarrow \bot }
        \infer2{\langle \rho_{a / b} (D) , \kappa \rangle \rightarrow D_1}
    \end{prooftree} {\small (w/e $D_1$ is $D$ renamed)} & \\


    \small
    \begin{prooftree}
        \hypo{ \langle D \{ T_0 \}, \kappa \rangle \rightarrow D}
        \infer1{\langle D \{ T_0, T_1, ..., T_n \} , \kappa \rangle \rightarrow D \{ T_1, ..., T_n \}}
    \end{prooftree} &

    \small
    \begin{prooftree}
        \hypo{ \langle D \{ T_0 \}, \kappa \rangle \rightarrow \bot}
        \infer1{\langle D \{ T_0, T_1, ..., T_n \} , \kappa \rangle \rightarrow \bot}
    \end{prooftree} & \\

  \bottomrule
\end{tabular}
\end{figure*}

\subsection{TQL: Syntax and Semantics}

By design, the TQL language is largely structured to serve as a surface syntax representation of the data discovery problem encoded within the algebraic language of \texttt{ImpRAT}. As such, TQL can be viewed as both a refined semantic subset of \texttt{ImpRAT}, restricted to avert syntactically correct but semantically nonsensical statements, and also an user-oriented extension of \texttt{ImpRAT}, with the inclusion of intuitive constructs and syntactic sugar. 

We provide the full syntax for TQL in BNF grammar through Figure ~\ref{fig:BNF-TQL}, and we define the semantics via translation to \texttt{ImpRAT} in Figure ~\ref{fig:SEM-TQL}. Specifically, we note that Figures ~\ref{fig:BNF-TQL} and ~\ref{fig:SEM-TQL} summarize definitions which are inductively constructed:
\begin{itemize}
    \item Figure ~\ref{fig:BNF-TQL} provides a constructive definition of TQL programs ($\texttt{prog}$) from statements ($\texttt{stmt}$) and relational expressions ($\texttt{expr}$). For example, the definition
    \begin{displaymath}
        \text{expr} ::= \text{expr} \text{ bop } \text{expr} {\:|\:} \text{expr} \texttt{[} \text{uop} \texttt{]} {\:|\:} \text{var} {\:|\:} \text{lit}
    \end{displaymath}
    states that a relational expression ($\texttt{expr}$) can be constructed either from the base cases of a dataset variable (\texttt{var}) or a dataset literal (\texttt{lit}), or inductively by applying binary and unary operators (\texttt{bop} and \texttt{uop}) to already-valid expressions.   
    \item Figure ~\ref{fig:SEM-TQL} provides a constructive definition of the translation function $\lBrack - \rBrack$, which maps a TQL program into its corresponding program in \texttt{ImpRAT}. For example, the translation
    \begin{displaymath}
        \lBrack \text{var} \texttt{:\{} \text{tp} \texttt{\}} \texttt{ = } \text{expr} \rBrack \triangleq x_{\text{var}} := \lBrack \text{expr} \rBrack; x_{\text{var}} \{ \lBrack \text{tp} \rBrack \}
    \end{displaymath}
    states that the syntactic construction $\lBrack \text{var} \texttt{:\{} \text{tp} \texttt{\}} \texttt{ = } \text{expr} \rBrack$ in TQL can be translated into the sequential statements $x_{\text{var}} := \lBrack \text{expr} \rBrack$ and $x_{\text{var}} \{ \lBrack \text{tp} \rBrack \}$ in \texttt{ImpRAT}, where $\lBrack \text{expr} \rBrack$ and $\lBrack \text{tp} \rBrack$ are translations for the sub-constructions of the initial syntactic construct. 
\end{itemize}
We now highlight a few elements of interest: 

Syntactically, the language constructs for TQL shown in Figure ~\ref{fig:BNF-TQL} were chosen to echo the user's intuition and to follow common practices existing across programming languages design. The set operations of relational algebra, UNION ($\cup$), DIFFERENCE ($-$) and CARTESIAN PRODUCT ($\times$), are encoded using the typical characters for computer arithmetic (\texttt{+}, \texttt{-}, \texttt{*}), matching the arithmetic conventions used for these set-theoretic operations. Comparator syntax for the atomic normal form propositions similarly follows the typical syntax for logical operands (e.g. \texttt{==}, \texttt{>=}, \texttt{||}, \texttt{\&\&}). Additionally, the brace notation (var \texttt{:\{} tp \texttt{\}}) used for type "checks" in TQL adopts similar style to the system of type annotations of the functional programming language OCaml. 

Semantically, we can view a program $\lBrack \text{prog} \rBrack$ in TQL as corresponding to a program $S$ in \texttt{ImpRAT}. To form a tuple $(S, \mathcal{D})$ corresponding to a complete data discovery program, we consider the set $\lBrack \mathcal{D} \rBrack$ of possible initial datasets to be the data repository associated with a query engine implementation. 

\begin{figure}
    \caption{Syntax for TQL}
    \label{fig:BNF-TQL}
    \begin{tabular}{c}
        \toprule

        \vbox{
            \begin{flalign*}
                \text{expr} &::= \text{expr} \text{ bop } \text{expr} {\:|\:} \text{expr} \texttt{[} \text{uop} \texttt{]} {\:|\:} \text{var} {\:|\:} \text{lit} \\
                \text{bop} &::= \texttt{+} {\:|\:} \texttt{-} {\:|\:} \texttt{*} \\
                \text{uop} &::= \text{attr} \texttt{->} \text{attr} {\:|\:} \text{pred} {\:|\:} \text{proj} \\
                \text{proj} &::= \text{attr} {\:|\:} \text{attr} \texttt{;} \text{proj} \\
                \text{pred} &::= \text{attr} \text{ cmp } \text{attr} {\:|\:} \text{attr} \text{ cmp } \text{val} {\:|\:} \texttt{!} \text{pred} \\
                &{\:|\:} \text{pred} \texttt{ \&\& } \text{pred} {\:|\:} \text{pred} \texttt{ || } \text{pred} \\
                \text{cmp} &::= \texttt{==} {\:|\:} \texttt{!=} {\:|\:} \texttt{>} {\:|\:} \texttt{>=} {\:|\:} \texttt{<} {\:|\:} \texttt{<=} \\
                \text{prog} &::= \text{stmt} {\:|\:} \text{stmt} \texttt{;} \text{prog} \\
                \text{stmt} &::= \text{var} {\:|\:} \text{var} \texttt{:\{} \text{tp} \texttt{\}} {\:|\:} \text{var} \texttt{ = } \text{expr} {\:|\:} \text{var} \texttt{:\{} \text{tp} \texttt{\}} \texttt{ = } \text{expr} \\
                &{\:|\:} \texttt{return } \text{expr} \\
                \text{tp} &::= \text{prp} {\:|\:} \text{prp} \texttt{;} \text{tp} \\
                \text{prp} &::= \texttt{\textbackslash /(} \text{pred} \texttt{)} {\:|\:} \texttt{/\textbackslash (} \text{pred}
                \texttt{)} {\:|\:} \texttt{[} \text{attr} \texttt{]}
            \end{flalign*}
        } \\

        \vbox{ 
            \begin{equation*}
                \text{var} ::= \text{string}
            \end{equation*}
            \begin{equation*}
                \text{attr} ::= \text{string literal(\texttt{'})} 
            \end{equation*}
            \begin{equation*}
                \text{lit} ::= \text{string literal(\texttt{"})} 
            \end{equation*}
            \begin{equation*}
                \text{val} ::= \text{floating-point numbers} 
            \end{equation*}
        } \\ 
            
        \bottomrule
    \end{tabular}
\end{figure}

\begin{figure}
  \caption{Semantic Translations for TQL}
  \label{fig:SEM-TQL}
  \def\arraystretch{1.5}
  \begin{tabular}{l}
    \toprule

    \vbox{
        \begin{flalign*}
            & \lBrack \text{stmt} \texttt{;} \text{prog} \rBrack \triangleq \lBrack \text{stmt} \rBrack ; \lBrack \text{prog} \rBrack \\
            & \lBrack \text{var} \rBrack \triangleq !x_{\text{var}} \\
            & \lBrack \text{var} \texttt{:\{} \text{tp} \texttt{\}} \rBrack \triangleq !x_{\text{var}}; x_{\text{var}} \{ \lBrack \text{tp} \rBrack \} \\
            & \lBrack \text{var} \texttt{ = } \text{expr} \rBrack \triangleq x_{\text{var}} := \lBrack \text{expr} \rBrack \\
            & \lBrack \text{var} \texttt{:\{} \text{tp} \texttt{\}} \texttt{ = } \text{expr} \rBrack \triangleq x_{\text{var}} := \lBrack \text{expr} \rBrack; x_{\text{var}} \{ \lBrack \text{tp} \rBrack \} \\
            & \lBrack \texttt{return } \text{expr} \rBrack \triangleq \Delta \lBrack \text{expr} \rBrack \\
            & \lBrack \text{expr} \text{ bop } \text{expr} \rBrack \triangleq \lBrack \text{expr} \rBrack \oplus \lBrack \text{expr} \rBrack \\
            & \lBrack \text{expr} \texttt{[} {\text{attr0} \texttt{->} \text{attr1}} \texttt{]} \rBrack \triangleq \rho_{a_0/a_1} (\lBrack \text{expr} \rBrack) \\
            & \lBrack \text{expr} \texttt{[} {\text{pred}} \texttt{]} \rBrack \triangleq \sigma_{\lBrack {\text{pred}} \rBrack} (\lBrack \text{expr} \rBrack) \\
            & \lBrack \text{expr} \texttt{[} {\text{attr1;...;attrN}} \texttt{]} \rBrack \triangleq \Pi_{a_1, ..., a_N} (\lBrack \text{expr} \rBrack) \\
            & \lBrack \text{prp} \texttt{; ... ;} \text{prp} \rBrack \triangleq \lBrack \text{prp} \rBrack, ..., \lBrack \text{prp} \rBrack \\
            & \lBrack \texttt{\textbackslash /(} \text{pred} \texttt{)} \rBrack \triangleq \exists \lBrack \text{pred} \rBrack \\
            & \lBrack \texttt{/\textbackslash (} \text{pred} \texttt{)} \rBrack \triangleq \forall \lBrack \text{pred} \rBrack \\
            & \lBrack \texttt{[} \text{attrN} \texttt{]} \rBrack \triangleq a_N \\
            & \lBrack \text{pred} \rBrack \triangleq \varphi
        \end{flalign*}
    } \\ 

  \bottomrule
\end{tabular}
\end{figure}

\section{Implementation Architecture}

In this section, we give an overview of the implementation architecture of our prototype TQL query engine, which solves the data discovery problem by computing a valid output on the TQL program representation of an \texttt{ImpRAT} expression. Since this problem is NP-Hard, our implementation focuses on producing a modular and generic prototype solver that can easily be modified and extended as needed. Our implementation can be found in the Appendix.

Recall from our discussion in Section 2 that the data discovery problem can be represented as a program formed from a tuple $(S, \mathcal{D})$, where $S$ is a program in \texttt{ImpRAT}, and $\mathcal{D}$ is our data repository of interest. Our architecture (see Figure ~\ref{fig:architecture}) consists of two main component groups corresponding to each element of this tuple. 

\begin{figure}
    \caption{Query Engine System Architecture}
    \label{fig:architecture}
    \begin{center}
        \includegraphics[width=\columnwidth]{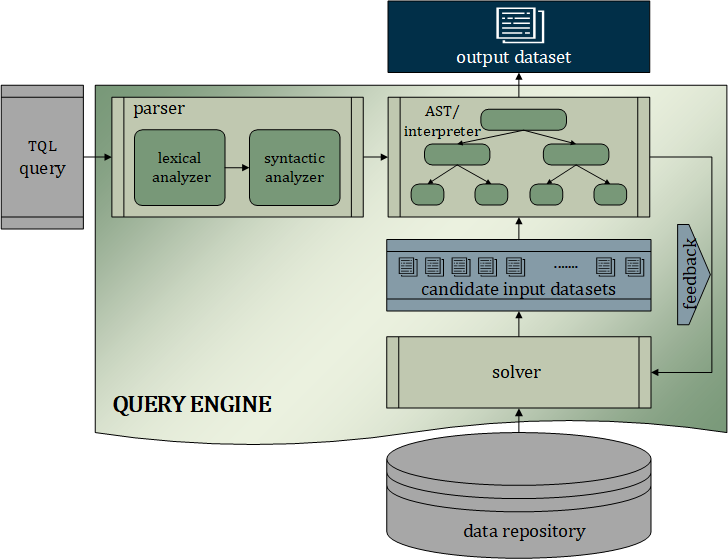}
    \end{center}
\end{figure}

\subsection{Parser-Interpreter}

We implement the lexical and syntactic analyzers for our TQL parser in Python using the Python Lex-Yacc (PLY) compilation tool developed by David M. Beazly \cite{ply}. The lex-yacc specifications employed for tokens and parsing is given through our BNF representation of TQL syntax in Figure ~\ref{fig:BNF-TQL}. 

The parser we build through PLY constructs the AST by a recursive process. Starting from primitive language constructs, the parser assembles complex TQL programs by converting each construct into an appropriate object in the \texttt{Node} superclass, and linking those objects into the desired tree structure. Interpretation and static analysis are both conducted on the tree via recursive calls to overloaded methods on the nodes of a tree structure. 

\subsection{Query Engine Solver}

At a high-level, the solver for our TQL query engine works -- taking inspiration from nondeterministic programming algorithms \cite{Abelson:structure} -- by repeatedly attempting to choose a candidate choice of input datasets for which a valid output dataset can be produced on the given query. More specifically, given an arbitrary data discovery query, the TQL query engine performs the following steps:
\begin{enumerate}
    \item The query engine calls a choice function -- a procedure which finds a candidate input on which the TQL statement might successfully evaluate. 
    \item The choice function uses static analysis on the AST and feedback from previous failed inputs (if any exist) to attempt to produce a valid candidate input.
    \item The candidate input is run on the AST.
    \item If the input succeeds, return the resultant dataset. Otherwise, repeat these steps with additional feedback for the choice function.
\end{enumerate}
 The intuition is to find a valid output dataset by enumerating over results produced by candidate inputs, leveraging feedback from failed inputs to achieve performance gains over naive enumeration. This procedure is given formally in Algo. ~\ref{alg:main}.

\begin{algorithm}
	\caption{Query Engine Algorithm}
	\label{alg:main}
	\begin{algorithmic}[1]
			\Require valid discovery-tuple $(S, \mathcal{D})$
			\Procedure{Solve}{$S, \mathcal{D}$}
                \State $F \gets \emptyset$ \Comment{Initialize feedback $F$}
                \While{true}
                    \State $I \gets \textsc{Choice}(S, \mathcal{D}, F)$ \Comment{Choose input $I$}
                    \State $O \gets \textsc{Run}(S, I)$ \Comment{Run program on $I$}
                    \If{$O \neq \bot$}
                        \State \Return $O$
                    \Else
                        \State $F \gets F \cup \{ I \}$
                    \EndIf
                \EndWhile
            \EndProcedure
	\end{algorithmic}
\end{algorithm}

For the implementation provided with this paper, we choose a classical backtracking algorithm as our choice function \cite{Abelson:structure}. We give it formally as Algo.~\ref{alg:choice}.

\begin{algorithm}
	\caption{Choice Function via Backtracking}
	\label{alg:choice}
	\begin{algorithmic}[1]
            \Require choice tuple $(S, \mathcal{D}, F)$
			\Function{Choice}{$S, \mathcal{D}, F$}
                \State $n \gets |S|$ \Comment{Compute number of inputs to $S$}
                \State $I \gets \emptyset$ \Comment{Initialize empty multiset $I$}
                \State \textbf{yield} \textsc{Backtrack}$(\mathcal{D}, n, I)$
            \EndFunction
            \Statex \vspace{-6pt}\hrulefill
			\Require search tuple $(\mathcal{D}, n, I)$
			\Function{Backtrack}{$\mathcal{D}, n, I$}
            \If{$|I| = n$}
                \State \Return $I$
            \EndIf
            \For{$D \in \mathcal{D}$}
                \State $I \gets I \cup \{ D\}$
                \State \textbf{yield} \textsc{Backtrack}$(\mathcal{D}, n, I)$
            \EndFor
            \EndFunction
	\end{algorithmic}
\end{algorithm}

\section{Next Steps}

This section highlights some ongoing work on the TQL system research project, and also discusses a few other potential research directions. 

\subsection{Choice Functions}
There are a multitude of approaches to constructing reasonable candidate choice functions. For example, it may be possible to develop probabilistic algorithms or to train machine learning models that exhibit strong performance in the general-case. 

Continuing the theme of this paper, however, one possible augmentation to the backtracking algorithm from Section 3 (see Algo ~\ref{alg:choice}) is to leverage type inference techniques from programming languages research to prune the search space of possible candidate inputs before interpreter evaluation. 

\subsection{Language Extensions}

We briefly outline a few potential language extensions for TQL:
\begin{itemize}
    \item \textbf{Advanced Join Operators}. TQL does include all primitive relational algebra operations, and can therefore simulate basic join and join-like operations. Including syntactic sugar would create a more accessible toolkit for the user. 
    \item \textbf{Complex Constraints}. The type constraints available in the TQL core language correspond to basic properties which are easy to verify within relational algebra. The introduction of additional constraints could greatly increase the practical expressiveness of the language. 
    \item \textbf{User Interface for Human Interaction}. Developing a user interface to interact with the discovery programming language can facilitate exploratory analysis, improve debugging workflows, and enable broader adoption for task-driven data discovery.

\end{itemize}





\section{Related Work}

In this section, we explore some related work on data discovery and on applications of programming languages research to domain-specific purposes. 

\textbf{Data Discovery Systems}. Much research has been conducted on the design and development of data discovery systems \cite{Paton:survey, huang:semiring, behme:fainder, gong:nexus2024, Galhotra:metam, Fernandez:aurum, tao:kyrix, cong:WrapGate}. Such research has included work on the development of efficient algorithms for a variety of critical data discovery problems, and implementing systems have spanned both language-based and more visually-oriented interfaces  \cite{Galhotra:metam, Fernandez:aurum, tao:kyrix}. We believe that TQL is the first general-purpose data discovery system to provide a domain-specific language with carefully formalized semantics that leverage modern programming languages results. 

\textbf{Algebraic Language Extensions}. Prior work in algebraic language extensions, such as Kleene Algebra with Tests (KAT) \cite{kozen:kat}, has explored the incorporation of boolean logic into existing algebraic systems. Such extended systems are valuable in a variety of contexts, including as semantic foundations for network programming languages \cite{anderson:netkat}, and for reasoning about imperative programs \cite{steffen:gkat}. The semantic foundations of TQL in \texttt{RAT} and \texttt{ImpRAT}, however, differ from systems like KAT in that they take the view of encoding predicate logic into types, rather than terms and control flow. 

\textbf{Type-driven Language Features}. Significant research in programming languages literature covers applications of type-based programming languages techniques to many other domains, including databases \cite{silva:typedb, Ohori:standardml, Caldwell:typechecking}. In most such research \cite{silva:typedb, Ohori:standardml, Caldwell:typechecking}, however, the primary focus of the type system is in providing formal guarantees of correctness. Although research into dependent type systems \cite{Pierce:advtype} has similar theoretical flavor to the proposed type system for TQL, there is no preceding work in which a type system was used to characterize the search space and drive data discovery.

\section{Conclusion}

This paper presents a novel general-purpose data discovery system through TQL, an extensible domain-specific language for data discovery designed to incorporate and extend techniques from modern programming languages research. In addition to syntax and semantics given via translation into the algebraic \texttt{ImpRAT} model, we also implement a generic and modular proof-of-concept solver for our language that can be easily extended to integrate both future language features and improved search algorithms.
\vspace{-3mm}


\bibliographystyle{ACM-Reference-Format}
\bibliography{sample-base}

\appendix
\vspace{-5mm}
\section{Data Availability}
The current version of the implementation can be found in the following GitHub repository: \url{https://github.com/andrew-kang/tql}.
\end{document}